\documentclass[12pt]{article}

\usepackage{amssymb,amsmath,amsfonts,eurosym,geometry,ulem,graphicx,caption,color,setspace,sectsty,comment,footmisc,caption,natbib,pdflscape,array,hyperref,tikz}
\usepackage{pgfplots}
\usepackage{amsmath, amsthm, hyperref, xcolor}

\usepackage{subcaption}

\newtheorem{lemma}{Lemma}

\newtheorem{proposition}{Proposition}
\newtheorem{proofp}{Proof of Proposition}
\newtheorem{lemmap}{Proof of Lemma}
\newtheorem{definition}{Definition}
\newtheorem{example}{Example}

\hypersetup{
    colorlinks=true,
    linkcolor=blue,
    citecolor=blue,
    urlcolor=blue
}

\newcolumntype{C}[1]{>{\centering\let\newline\\arraybackslash\hspace{0pt}}m{#1}}
\newcolumntype{R}[1]{>{\raggedleft\let\newline\\arraybackslash\hspace{0pt}}m{#1}}

\begin{document}

\begin{titlepage}
\title{Third Degree Price Discrimination Under Costly Information Acquisition\thanks{I would like to express my gratitude to Yuhta Ishii, Ron Siegel and Nima Haghpanah for their valuable feedback.}}
\author{Irfan Tekdir\thanks{Department of Economics, Pennsylvania State University}}
\date{\today}
\maketitle
\begin{abstract}
This paper investigates third-degree price discrimination under endogenous market segmentation. Segmenting a market requires access to information about consumers, and this information comes with a cost. I explore the trade-offs between the benefits of segmentation and the costs of information acquisition, revealing a non-monotonic relationship between consumer surplus and the cost of information acquisition for monopolist. I show that in some markets, allowing the monopolist easier access to customer data can also benefit customers. I also analyzed how social welfare reacts to changes in the cost level of information acquisition and showed that the non-monotonicity result is also valid in social welfare analysis.
\vspace{2in}

\vspace{0.1in}

\bigskip
\end{abstract}
\setcounter{page}{0}
\thispagestyle{empty}
\end{titlepage}
\pagebreak \newpage

\doublespacing

\section{Introduction} \label{sec:introduction}
Price discrimination is an important agenda item for many industries. There are nearly endless techniques employed by companies for implementing price discrimination. One of the most commonly used techniques is third-degree price discrimination, which involves charging different prices to different groups of people. Sometimes, this type of price discrimination is easy to observe, such as when movie theaters offer discounted tickets to students. However, it can also be more subtle, such as when online retailers track customers' browsing histories using cookies and then charge different prices to different customers based on that information. This is more hidden because consumers may not be aware that prices are tailored based on their browsing history. Even if they are aware of this, they may not know the extent of the price discrimination. The key point for third-degree price discrimination is its heavy reliance on accessing information about consumer tastes.

Companies collect extensive amounts of information about their customers, segment the market based on this information, and charge different segments different prices. By doing so, they aim to extract more surplus from the market. However, gathering information about customers is not free. Companies incur significant costs for that. Another cost involved is processing that information. Amazon, Google, and Meta collect extensive user data to develop their products, and this data allows them to use pricing strategies. However, processing this data requires skilled data analysts/scientists and significant computing power, which comes at a cost to the company.

Moreover, recently data privacy concerns have increased. In 2018, there was a big scandal involving a political consulting firm named Cambridge Analytica which improperly exploited data from millions of Facebook profiles without their consent for their political campaigns.
To address increased concerns around data privacy, Apple enforced App Tracking Transparency(ATT) in April 2021. The European Union's General Data Protection Regulation (GDPR) restricted firms' data collection, storage, and usage by imposing new rules. All of these changes made it harder for companies like Facebook to access detailed consumer data. Therefore understanding how the increased cost of information acquisition affects markets from both the consumer's perspective and the efficiency perspective is important for policymakers.

 This paper explores the relationship between consumer welfare and the cost of information acquisition in monopolistic markets. \cite{BBM}--henceforth, BBM--  analyzed all (Consumer Surplus, Producer Surplus) pairs that are achievable through some segmentation. They introduced a \textit{surplus triangle} and demonstrated that any point within the triangle can be achieved through some segmentation. Additionally, they provided methods for constructing such segmentation. However, in their paper, segmentation does not involve a cost. In reality, segmentation requires information, and obtaining that information comes with a cost. The monopolist needs to choose the optimal amount of information to gather about her customers because, although she gains benefits by segmenting the market, she also incurs a cost for doing so.
This paper fills this gap and aims to determine the optimal segmentation in a market when the information acquisition is costly (\hyperref[prop:1]{Proposition \ref*{prop:1}}).(\hyperref[prop:3]{Proposition \ref*{prop:3}}) explains how consumer surplus responds to change in the cost of information, and (\hyperref[prop:4]{Proposition \ref*{prop:4}}) explains how efficiency in the market changes with the change in cost of information.

To address these questions, I developed a theoretical model based on BBM, in which I used a cost function that is proportional to the expected reduction in entropy \cite{shannon} to measure the cost of information acquisition. My results indicate that there is a \textit{non-monotonic} relationship between consumer surplus and the level of cost of information acquisition in some markets. This result might seem counterintuitive at first glance because one might think that if the monopolist accesses customer information at a lower cost, the monopolist would extract more surplus from the consumers, and the relationship between consumer surplus and cost level would be monotonic. However, my results show that in some markets, allowing the monopolist easier access to customer data can also benefit customers.

A simple intuition can help us understand this counterintuitive result better. Let's assume there are only two types of consumers: low-types and high-types, with more high-types than low-types. Initially, it's optimal for the monopolist to charge a high price before any segmentation.
If gathering information is costless for the monopolist, she would collect all the necessary data and engage in perfect price discrimination under complete information, leaving no consumer surplus. 

Conversely, if the cost of gathering information is too high, the monopolist would choose not to engage in any segmentation activity and stick with the high price. This would exclude some low-type consumers from the market transaction, and since high-types are paying the high price, the consumer surplus would again be zero.
However, when the cost of information acquisition is at a moderate level, the monopolist can create a segment where charging a low price is optimal. By doing so, she can separate the two types of consumers. In the low-price segment, there will be more low-type consumers enjoying the market transaction, and some high-type consumers will also be present in that segment because the monopolist was unable to separate the two types perfectly. These high-type consumers will start to enjoy some surplus in the low-price segment.
As the cost of information acquisition increases, the size of the low-price segment shrinks, and the overall consumer surplus decreases. This result is highly important for policymakers because it shows that under certain conditions, when the monopolist has more information about their customers, they can create tailored segments, which can benefit consumers. Therefore, it's crucial to understand that not allowing access to some consumers' preferences does not always help consumers.

Another contribution of this paper is the rationalization of the surplus triangle introduced in BBM. It is shown in BBM that any point in the surplus triangle can be achieved via some segmentation, and they demonstrate the method of constructing a segmentation that yields a particular (CS, PS) pair in the triangle. I define  a point in the surplus triangle as rationalizable if it can result from optimal market segmentation for a monopolist when the segmentation is endogenous. I prove that any interior point of surplus triangle is rationalizable in the binary case.
\subsection{Related Literature} \label{subsec:literature}
The purpose of this literature review is to provide a comprehensive overview of the existing research on third-degree price discrimination, particularly in the context of costly information acquisition. In economics, third-degree price discrimination is widely studied in the context of monopolistic and oligopolistic markets. \cite{pigou1920economics} was the first person who classified the types of price discrimination. A lot of researchers followed this paper and anayzed the applications of price discrimination. \cite{varian1989price} provides a nice survey about the economical applications of price dicriminations. Varian highlights that the most common form of price discrimination is likely third-degree price discrimination.

\cite{robinson1933economics} was the first to analyze the welfare effects of third-degree price discrimination. Most of the papers in the old literature about the welfare effects of price discrimination assumed that there is no cost associated with it, \cite{leeson2008cost} showed that when the costliness of price discrimination is taken into consderation, perfect price discrimination is generally socially inefficient for monopolists. The most important difference between the existing literature and my paper is that these studies of price discrimination that analyze the welfare effects of price discrimination take market segmentation as exogenous. However, in my model, segmentation is endogenous.

Recent studies of third-degree price discrimination have also explored the nuances in the context of monopolistic markets, digital markets, and two-sided platforms. \cite{BBM} analyzed the limits of price discrimination, showing all achievable welfare points in the monopolistics market.\cite{lu2019endogenous} analyzed endogenous third degree price discrimination in Hotelling model. \cite{bergemann2022third} compares third-degree price discrimination with uniform pricing, demonstrating that a uniform price can guarantee attaining half of the maximum profits achievable under a monopoly with concave profit functions. The application of third-degree price discrimination in two-sided markets is studied by \cite{de2023third}, who show that price discrimination can be beneficial for both sides of the markets and increase total welfare.

The applications of costly information acquisition represent another growing literature, and this paper can be considered one of its examples. \cite{matyskova2021bayesian} models costly information acquisition in a sender and receiver game. My paper can be seen as an application of this paper, where the monopolist acts as the receiver and consumers act as the sellers, with only the receiver paying the cost for the information. \cite{thereze2022screening} studies the screening problem under costly information. He shows that there is a non-monotonic relationship between information cost and the profit level of the menu designer. The key difference between my paper and that paper, except that one is studied in third-degree price discrimination and the other is studied in a screening setup, is that in my paper, the monopolist incurs the cost to acquire information about consumers, and in his paper, the consumer pays the information cost to acquire information about his valuations for the menu items that he is offered.

My paper also contributes to the literature on Bayesian privacy. \cite{eilat2021bayesian} studies the relationship between the privacy levels and welfare in the market. Both costly information acquisition and maintaining privacy are about making it harder for firms to access consumer data. The important difference is that in my paper, I allow any type of information if the monopolist is willing to pay the cost of that information. However, privacy imposes a restriction and states that even if you have the money to acquire information, you cannot obtain certain amounts of data due to privacy rights. My results combined with the results of this paper can provide good insight to policy makers about the potential trade-offs between privacy regulations and consumer welfare or market efficiency.

 For technical tools, I must mention important papers from the Rational Inattention literature, such as \cite{caplin_dean}, \cite{matejka}, and \cite{caplin2019rational}, which have significantly helped me.

The organization of the rest of the paper is as follows:
Section 2 introduces the model and provides the definition of optimal segmentation. Section 3 explains how to solve the monopolist's problem using two different techniques: the Concavification technique \cite{aumann1995repeated} from the Bayesian Persuasion \cite{Bayesian_Persuasion} literature, which provides mostly geometric intuition about the optimal segmentation in the binary case, and a more technical solution method from the Rational Inattention literature for more general cases. Section 4 provides a welfare analysis considering the cost of information. Section 5 discusses the rationalization of the surplus triangle introduced in BBM. Section 6 concludes the paper, and Section 7 provides an Appendix with the proofs of lemmas and propositions throughout the paper.
\section{Model} \label{sec:model}

There is a monopolist producing a single good in the market and a continuum of consumers with a total mass of 1. Each consumer has a valuation for the good from a finite valuation set $\Omega = \{\omega_1, \omega_2, \ldots, \omega_K\}$, where it is assumed that $0 < \omega_1 < \omega_2 < \ldots < \omega_K.$ Generic element in $\Omega$ is denoted by $\omega.$ Seller has zero marginal cost to produce the good.

A \textit{market} is a vector $\mu=(\mu_1,\mu_2,...,\mu_K)$ where $\mu_k$ is the ratio of consumers whose valuation for the good is equal to $\omega_k$. Let $\mathcal{M}$ denote the set of all possible markets.
\begin{equation*}
    \mathcal{M}\equiv \{\mu \in \mathbb{R}^K_{+}|\displaystyle\sum_{k=1}^{K}\mu_k=1\}
\end{equation*}

Suppose an \textit{aggregate market} is fixed and denoted as $\mu^* \in \mathcal{M}$. This can be interpreted as the monopolist's prior belief about the distribution of consumer valuations in the market. Knowing $\mu^* = (\mu_1^*, \mu_2^*, \ldots, \mu_K^*)$ means that if the monopolist randomly encounters a consumer from the population, the probability that this consumer's valuation is $\omega_k$ is equal to $\mu_k^*$. 

Let $s(p,\omega)$ and $b(p,\omega)$  be the utility of the monopolist(seller) and utility of type $\omega$ consumer(buyer) when the monopolist charges price p for that type in any market.

\begin{align*}
s(p,\omega) &= 
  \begin{cases} 
    p & \text{if } \omega \geq p, \\
    0 & \text{otherwise},
  \end{cases} \quad
b(p,\omega)= 
\begin{cases} 
    \omega- p & \text{if } \omega \geq p, \\
    0 & \text{otherwise},
  \end{cases}
\end{align*}

 It is assumed that the monopolist charges the optimal price, \( p^*(\mu) \) \footnote{Note that this price represents the optimal uniform monopoly pricing in the market $\mu$. Also, it is easy to see that the optimal price should be an element of $\Omega$.}
, in any market \( \mu \). That is,

\[
p^{*}(\mu) \in \arg\max_p \sum_{k=1}^{K} s(p,\omega_{k}) \mu_{k}
\]

Let $\mathcal{M}_k$ be the set of markets in which charging price $\omega_k$ is optimal for the monopolist. That is
\begin{equation*}
     \mathcal{M}_k\equiv \{\mu \in \mathcal{M}|\omega_k\in \arg\max_p \sum_{k=1}^{K} s(p,\omega_{k}) \mu_{k}\}
\end{equation*}

\begin{equation*}
    \mathcal{M} \equiv \bigcup_{k=1}^{K} \mathcal{M}_k
\end{equation*}

\begin{example}
Suppose $K=2$, $\Omega=\{1,2\}$, and $\mu_{2}^*=0.6$ represents the ratio of type two consumers. In this scenario, a market with two types can be represented using a one-dimensional simplex, which is essentially a simple line. If $\mu_{2}^* \leq 0.5$, then charging \$1 is optimal for the monopolist and we are in $\mathcal{M}_1$; if $\mu_{2}^* \geq 0.5$, then charging \$2 is optimal for the monopolist and we are in $\mathcal{M}_2$. \hyperref[fig:market]{Figure \ref*{fig:market}} illustrates the market representation for two types.
\end{example}

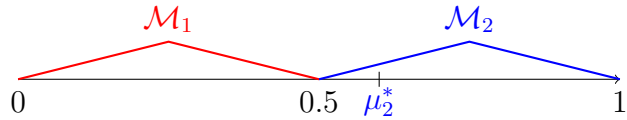
\begin{figure}[h]
    \centering
    \begin{tikzpicture}
    \draw[->] (0,0) -- (8,0) node[anchor=north] {};
    
    \draw[thick, red] (0,0) -- (2,0.5) -- (4,0);
    
    \draw[thick, blue] (4,0) -- (6,0.5) -- (8,0);

    \node at (2,0.8) {\textcolor{red}{$\mathcal{M}_1$}};
    \node at (6,0.8) {\textcolor{blue}{$\mathcal{M}_2$}};
    
    \node at (0,-0.3) {0};
    \node at (4,-0.3) {0.5};
    \node at (8,-0.3) {1};
    \node at (4.8,-0.3) {\textcolor{blue}{$\mu_{2}^*$}};
    \draw[dashed] (4.8,0) -- (4.8,-0.1);
    \draw[dashed] (4.8,0) -- (4.8,0.1);
\end{tikzpicture}
    \caption{Market representation on a simple line}
    \label{fig:market}
\end{figure}

  \textit{Segmentation} $\tau$ of an aggregate market $\mu^*$ can be considered as probability distribution over $\mathcal{M}$ in which $\tau(\mu^{s})$ is the proportion of consumers in market $\mu^{s}$ \footnote{ Note that a generic element of $\text{supp}(\tau)$ is denoted by $\mu^s$, representing any \textit{segment} in $\tau.$} such that $\displaystyle\sum_{\mu^{s} \in supp (\tau)}\tau(\mu^{s})\mu^{s}=\mu^*$. The set of all possible segmentation is equal to \\
 \begin{equation*}
     \mathcal{T}=\{\tau \in \Delta(\mathcal{M})|\displaystyle\sum_{\mu \in supp (\tau)}\tau(\mu^{s})\mu^{s}=\mu^*, |supp(\tau)| < \infty\}
 \end{equation*} 

In simple terms, segmentation involves accessing additional information about customers, categorizing them into different segments, and tailoring prices for each specific segment.

Segmentation strategies share a common feature: the need to access additional information about  customers, which often comes at a cost. Whether it's purchasing cookies, entering into data-sharing agreements, or creating a platform to collect more data, each method incurs some cost for a monopolist.

I will model the cost of information acquisition for the monopolist using an Information Theoretic approach. Initially, the monopolist had some uncertainty about her customers in the aggregate market $\mu^{*}$, and she gathered some information to segment this aggregate market. She still has some uncertainty for each segment that she created, but she has overall less uncertainty compared to the beginning. To reduce her uncertainty, she needed to incur some costs to acquire additional information about her consumers' tastes. I will assume a posterior-separable cost function by following \cite{caplin_dean}, which will be in the following format:
\begin{equation*}
c(\tau;\mu^{*},k) = k\left[H(\mu^{*}) - \mathbb{E}_{\tau}(H(\mu^s))\right]
\end{equation*}

where $H: \Delta(\Omega) \rightarrow \mathbb{R}$ is assumed as strictly concave and continuous function. Here, $k$ denotes the cost parameter that ranges from 0 to infinity. When $k$ equals 0, the monopolist will always employ perfect price discrimination, and when $k$ goes to infinity, the monopolist will approach to uniform monopoly pricing in the aggregate market and will not segment the market due to the high cost. However, as $k$ varies, the monopolist needs to decide on the optimal level of information acquisition, as segmentation brings benefits but still incurs costs. I will specifically focus on Shannon entropy, \cite{shannon} where $H(\mu)=-\sum_{k=1}^{K}\mu_k\ln(\mu_{k})$ satisfies my two general assumptions, and it is assumed by convention that $0\ln 0=0$. The cost of information or the cost of segmentation will be modeled as proportional to expected reduction in entropy as it is done in \cite{Costly_Persuasion} and in most of the recent papers. 

The monopolist has to make two decisions to maximize her profit: the first involves charging the optimal uniform monopoly pricing, $p^{*}(\mu^{s})$, for each segment $\mu^{s}$, and the second involves choosing the optimal level of information acquisition, given the optimal pricing decision and cost of information acquisition. To determine the optimal level of information acquisition, the monopolist solves the following maximization problem:

\begin{align*}
\max_{\tau \in \mathcal{T}} \quad & \mathbb{E}_{\tau}\left[\underbrace{\sum_{k=1}^{K} s(p^{*}(\mu^{s}), \omega_{k}) \mu^{s}_{k} + kH(\mu^s)}_{\text{Let's call this part } v(\mu^{s}):\text{the net utility of having segment $\mu^{s}$}}\right] - kH(\mu^{*})
\end{align*}

While there are infinitely many ways to segment the market, the subsequent lemma narrows our focus to a finite number of possible segmentations. This will help us when we characterize the optimal segmentation.

 \begin{lemma} Let $\tau^*$ be the optimal segmentation for monopolist, then $|supp(\tau^*)| \leq |\Omega|$.
 \label{lemma1}
\end{lemma}

\begin{proof}

Detailed proof is in Appendix

\end{proof}

 We are ready to give a full definition of optimal segmentation.

\begin{definition}
    The \textit{optimal segmentation is a tuple, ($\tau^{*}$,$p^{*}(.))$}, that satisfies the followings:
    
    \begin{enumerate}
        \item $p^{*}(.)$ solves $\max_p \sum_{k=1}^{K} s(p,\omega_{k}) \mu^{s}_{k}$ for each $\mu^{s} \in supp (\tau^{*}).$ 
        \item Given $p^{*}(.)$ and $c(\tau;\mu^{*},k)$, $\tau^{*}$ solves
        \begin{align*}
         \max_{\tau \in  \mathcal{T}} \quad & \mathbb{E}_{\tau}[ v(\mu^{s})]
    \end{align*}
     \item $|supp(\tau^*)| \leq |\Omega|$
    \end{enumerate}
    \label{def:definition}
\end{definition}

In \hyperref[lemma1]{Lemma \ref*{lemma1}} , there is no mention of the uniqueness of the optimal segmentation. To address this, I will assume that if the monopolist considers two segmentations optimal for herself, she will prefer the one that consumers would prefer. Therefore, I will focus on the consumer-preferred optimal segmentation for the monopolist. From now on I will use $\tau^{*}$ to refer the optimal segmentation to ensure clarity in reference.

Given the optimal segmentation $\tau^{*}$ in the aggregate market $\mu^*$, we can define the \textit{Consumer Surplus ($CS)$}, \textit{Producer Surplus ($PS)$}, and \textit{Total Surplus ($TS)$} in the optimal segmentation as follows:

\begin{align*}
         CS &= \sum_{\mu^{s} \in \text{supp}(\tau^{*})} \tau^{*}(\mu^{s}) [\sum_{j=1}^{K} b(p^{*}(\mu^{s}),\omega_{j})\mu^{s}_{j}], \\
      PS &= \sum_{\mu^{s} \in \text{supp}(\tau^{*})} \tau^{*}(\mu^{s}) [v(\mu^{s})]\\
     TS &=\sum_{\mu^{s} \in \text{supp}(\tau^{*})} \tau^{*}(\mu^{s}) [v(\mu^{s})+\sum_{j=1}^{K}b(p^{*}(\mu^{s}),\omega_{j}))\mu^{s}_{j}]
\end{align*}

By following \hyperref[def:definition]{Lemma \ref*{def:definition}}, we can state the monopolist's problem as follows:
\begin{equation*}
  \begin{aligned}
         \displaystyle\max_{\{\mu^{1},\mu^{2},..,\mu^{K}\}} \quad &\tau(\mu^{1})v(\mu^{1})+\tau(\mu^{2})v(\mu^{2})+..+\tau(\mu^{K})v(\mu^{K})& \\
       \text{subject to} &
      \displaystyle\sum_{j=1}^{K}\tau(\mu^{j})\mu^{j}=\mu^* ,
      \displaystyle\sum_{j=1}^{K}\tau(\mu^{j})=1 , \quad \forall j, \tau(\mu^{j})  \geq 0 .
  \end{aligned}
\end{equation*}

\section{Solving Monopolist’s Problem} \label{sec:problem}

To solve the monopolist's problem, we can adapt the following proposition from the \cite{caplin_dean} paper, which provides an interior solution to our problem.
\begin{proposition}
\label{prop:1}
Let $\mu_{i}^{j}$ be the ratio of type $i$ consumers in segment $j$. Given $\Omega$ and $\mu^{*}$, the optimal segmentation $\tau^{*}$ satisfies the following conditions:

\begin{enumerate}
    \item \textbf{Bayes Plausibility Condition}: 
    \[
    \displaystyle\sum_{\mu \in \mathrm{supp}(\tau^{*})}\tau^{*}(\mu^{s})\mu^{s}=\mu^*
    \]

    \item \textbf{Invariant Likelihood Ratio Equations for Chosen Segments}: 
    Given $\mu^{j},\mu^{m} \in \mathrm{supp}(\tau^{*})$ and for any $\omega_{i} \in \Omega$,
    \[
    \frac{\mu_{i}^{j}}{e^{\frac{s(p_{j},\omega_{i})}{k}}}=\frac{\mu_{i}^{m}}{e^{\frac{s(p_{m},\omega_{i})}{k}}}
    \]

    \item \textbf{Likelihood Ratio Inequalities for Not-Chosen Segments}: 
    Given $\mu^{j}\in \mathrm{supp}(\tau^{*})$ and  $\mu^{t} \notin \mathrm{supp}(\tau^{*})$,
    \[
    \sum_{i=1}^{K}\left[\frac{\mu_{i}^{j}}{e^{\frac{s(p_{j},\omega_{i})}{k}}}\right]e^{\frac{s(p_{t},\omega_{i})}{k}} \leq 1
    \]
\end{enumerate}
\end{proposition}

\begin{proof}
 Detailed proof is in Appendix.
\end{proof}

The proposition states that optimal segmentation should first satisfy the Bayes plausibility conditions, which are captured by default in the definition of segmentation. The second condition indicates that when the monopolist decides on a segment, her decision about which fraction of each type should be included in which segment is influenced by the monopolist's utility derived from those types within the segments and the third condition is about for non-chosen segments.

Overall, Proposition 1 provides the solution for optimal segmentation in the market for a given specific cost function. However, it doesn't fully characterize the optimal market segmentation. For example, Proposition 1 is silent about the uniqueness of the optimal solution and does not provide any conditions on it. We will see that the non-uniqueness of optimal points is not an issue in the binary case, but for the general case, one needs to impose further restrictions on $\Omega$ to ensure a unique solution.

With an increased level of cost, the monopolist may find that segmenting the market is too costly and may prefer to follow a uniform monopoly pricing rule. The gain from segmentation can be defined as the monopolist's expected profit after segmentation minus the expected profit from the aggregate market. If the gain from segmentation is positive, the monopolist finds it optimal to segment the market. In some markets, due to the high level of segmentation, it is not beneficial to segment the market. However, the following proposition states that there are some special markets in which the monopolist always chooses to segment regardless of the cost level.
\begin{proposition}
\label{prop:2}
When the monopolist is indifferent between charging two different prices in the aggregate market $\mu^{*}$, that is $\mu^{*} \in \mathcal{M}_r \cap \mathcal{M}_s$ for some $r<s$, then she always chooses to segment the market regardless of the cost level.  
\end{proposition}
\begin{proof}
Detailed proof is in Appendix.
\end{proof}

The intuition behind Proposition 2 is as follows: When you are indifferent between two actions—meaning both actions give the same utility and you are unsure which action to choose—then, regardless of the cost, you should acquire the information that helps you make your decision with more precision. The proof is based on the idea that a small amount of information may be more beneficial than its cost when you are indifferent between two actions. Proposition 2 states that there will always be such a small level of information.
\subsection{Binary Case, $|\Omega|=2$}
In the binary case where $\Omega= \{\omega_1, \omega_2\}$ with $\omega_1 < \omega_2$, we can employ two different tools from two distinct literatures. The first tool is derived from the \textit{Rational Inattention} approach, as detailed in \hyperref[prop:1]{Proposition \ref*{prop:1}}, to identify the optimal segments in the market. The second tool, known as the \textit{Concavification} approach, comes from the \textit{Bayesian Persuasion} literature. While the first tool can be used for general cases, the Concavification method is particularly useful for solving binary action problems.
\subsubsection{Rational Inattention Approach}
With two types of consumers, the monopolist either chooses not to segment the market due to the high cost of information or segments the market into two segments. The optimal segments, $\mu^1$ and $\mu^2$, would satisfy the following equations, which are necessary and sufficient conditions for optimality:

 \[
  \frac{\mu_{1}^{1}}{e^{\frac{\omega_1}{k}}}= \frac{\mu_{1}^{2}}{1}
    \]

\[
    \frac{\mu_{2}^{1}}{e^{\frac{\omega_1}{k}}}= \frac{\mu_{2}^{2}}{e^{\frac{\omega_2}{k}}}
    \]
These equations imply that the ratio of corresponding types in the optimal segments is adjusted according to their valuations. Specifically, type-1 consumers will buy the good in segment 1 but not in segment 2, making it more likely that the monopolist will place type-1 consumers in segment 1. On the other hand, type-2 consumers will pay $\omega_1$ in segment 1 but $\omega_2$ in segment 2, making it more likely that the monopolist will place type-2 consumers in segment 2.

In order to find the optimal segmentation, we also need to exploit Bayes Plausibility condition. By doing so, we obtain a unique solution for the optimal segmentation in the binary case and then we can have welfare effect of costly information acquisition as it is discussed in \hyperref[analysis]{Section \ref*{sec:analysis}}.

\subsubsection{Concavification Approach}

In the binary case, there is a powerful tool called the Concavification technique, which was first introduced in \cite{aumann1995repeated} and gained popularity with the \cite{Bayesian_Persuasion} paper. The technique is all about splitting the aggregate market $\mu^{*}$ into segments $\mu^{s}$ in a way that maximizes the expectation of net utilities from segments. For a given $\mu^{*}$, the maximum occurs at the smallest concave function over the net utility functions when it is calculated on $\mu^{*}$.

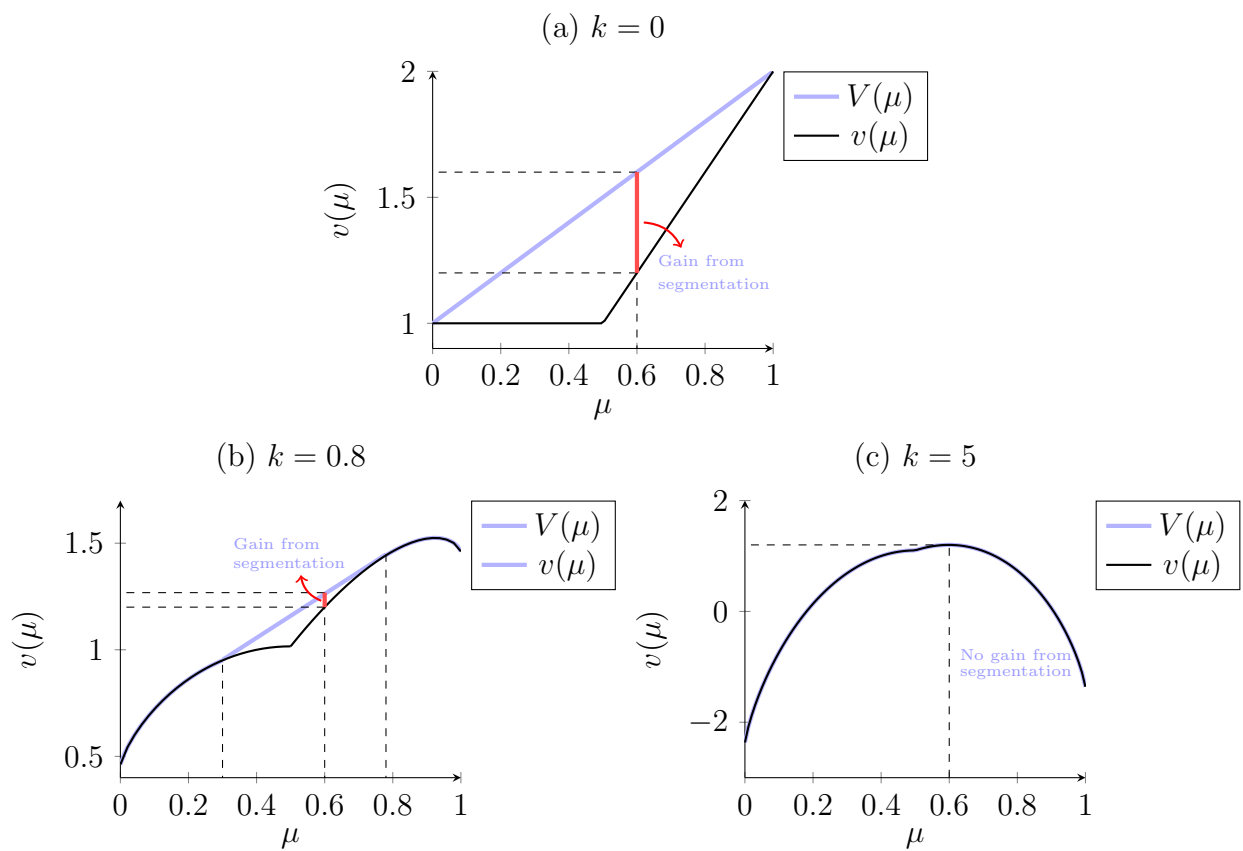
\begin{figure}[h]
    \centering
    \begin{tikzpicture}
        \begin{axis}[
            width=0.37\textwidth,
            title={(a) $k = 0$},
            xlabel={$\mu$},
            ylabel={$v(\mu)$},
            xmin=0, xmax=1,
            ymin=0.9, ymax=2,
            grid style={line width=.1pt, draw=gray!10},
            major grid style={line width=.2pt,draw=gray!50},
            axis lines=left,
            legend pos=outer north east
        ]
            \addplot[
                domain=0:1,
                samples=100,
                ultra thick,
                blue!30
            ]
            {1+x};
    \addlegendentry{$V(\mu)$}
            
            \addplot[
                domain=0:1,
                samples=100,
                thick,
                black
            ]
            {x <= 0.5 ? 1 : 2*x};
        \addlegendentry{$v(\mu)$}
            
             \draw [dashed, black] (axis cs:0.6, 0) -- (axis cs:0.6, 1.6);
            \draw [dashed, black] (axis cs:0.6, 1.2) -- (axis cs:0, 1.2);
            \draw [dashed, black] (axis cs:0.6, 1.6) -- (axis cs:0, 1.6);
        \draw [ultra thick, red!70] (axis cs:0.6, 1.2) -- (axis cs:0.6, 1.6);
    \node [blue!40, right] at (axis cs:0.63, 1.25) {\tiny Gain from };
    \node [blue!40, right] at (axis cs:0.63, 1.15) {\tiny segmentation };
         \draw [red, thick, ->, bend left] (axis cs:0.62, 1.4) to (axis cs:0.73, 1.3);
        \end{axis}
    \end{tikzpicture}

\begin{tikzpicture}
    \begin{axis}[
        width=0.37\textwidth,
        title={(b) $k = 0.8$},
        xlabel={$\mu$},
        ylabel={$v(\mu)$},
        xmin=0, xmax=1,
        ymin=0.4, ymax=1.7,
        grid style={line width=.1pt, draw=gray!10},
        major grid style={line width=.2pt,draw=gray!50},
        axis lines=left,
        legend pos=outer north east
    ]

    \addplot[
        domain=0:0.3,
        ultra thick,
        blue!30
    ] {1 - 0.8*(x*ln(x) + (1-x)*ln(1-x) - 0.6*ln(0.6) - 0.4*ln(0.4))};

    \addplot[
        domain=0.3:0.78,
        ultra thick,
        blue!30
    ] {0.644 + 1.03*x};

    \addplot[
        domain=0.78:1,
        ultra thick,
        blue!30
    ] {2*x - 0.8*(x*ln(x) + (1-x)*ln(1-x) - 0.6*ln(0.6) - 0.4*ln(0.4))};
    \addlegendentry{$V(\mu)$}

    \addplot[
        domain=0:0.5,
        thick,
        black
    ] {1 - 0.8*(x*ln(x) + (1-x)*ln(1-x) - 0.6*ln(0.6) - 0.4*ln(0.4))};

    \addplot[
        domain=0.5:1,
        thick,
        black
    ] {2*x - 0.8*(x*ln(x) + (1-x)*ln(1-x) - 0.6*ln(0.6) - 0.4*ln(0.4))};
    \addlegendentry{$v(\mu)$}
    \draw [dashed, black] (axis cs:0.3, 0) -- (axis cs:0.3, 0.953);
    \draw [dashed, black] (axis cs:0.78, 0) -- (axis cs:0.78, 1.4474);

    \draw [dashed, black] (axis cs:0.6, 0) -- (axis cs:0.6, 1.268);
    \draw [dashed, black] (axis cs:0.6, 1.2) -- (axis cs:0, 1.2);
    \draw [dashed, black] (axis cs:0.6, 1.268) -- (axis cs:0, 1.268);
     \draw [ultra thick, red!70] (axis cs:0.6, 1.2) -- (axis cs:0.6, 1.268);

      \node [blue!40, right] at (axis cs:0.3, 1.5) {\tiny Gain from };
    \node [blue!40, right] at (axis cs:0.3, 1.4) {\tiny segmentation };
         \draw [red, thick, ->, bend left] (axis cs:0.59, 1.23) to (axis cs:0.53, 1.35);

    \end{axis}
\end{tikzpicture}
    \begin{tikzpicture}
        \begin{axis}[
            width=0.37\textwidth,
            title={(c) $k = 5$},
            xlabel={$\mu$},
            ylabel={$v(\mu)$},
            xmin=0, xmax=1,
            ymin=-3, ymax=2,
            grid style={line width=.1pt, draw=gray!10},
            major grid style={line width=.2pt,draw=gray!50},
            axis lines=left,
            legend pos=outer north east
        ]
            \addplot[
                domain=0:1,
                samples=100,
                ultra thick,
                blue!30
            ]
            {max(1 - 5*(x*ln(x) + (1-x)*ln(1-x)-0.6*ln(0.6)-0.4*ln(0.4)),2*x - 5*(x*ln(x) + (1-x)*ln(1-x)-0.6*ln(0.6)-0.4*ln(0.4)))};
            \addlegendentry{$V(\mu)$}
            
            \addplot[
                domain=0:1,
                samples=100,
                thick,
                black
            ]
            {x <= 0.5 ? 1 - 5*(x*ln(x) + (1-x)*ln(1-x)-0.6*ln(0.6)-0.4*ln(0.4)): 2*x - 5*(x*ln(x) + (1-x)*ln(1-x)-0.6*ln(0.6)-0.4*ln(0.4))};
            \addlegendentry{$v(\mu)$}
            
           \draw [dashed, black] (axis cs:0.6, -3) -- (axis cs:0.6, 1.2);
            \draw [dashed, black] (axis cs:0.6, 1.2) -- (axis cs:0, 1.2);

    \node [blue!40, right] at (axis cs:0.6, -0.8) {\tiny No gain from };
    \node [blue!40, right] at (axis cs:0.6, -1.1) {\tiny segmentation };
        \end{axis}
    \end{tikzpicture}
    \caption{Monopolist' gain from segmentation with different $k$}
    \label{fig:concavification}
\end{figure}

Let $V(\mu)$ be the smallest concave function over net utility functions $v(\mu)$. As seen in Figure 2, when $k=0$, optimal segments occur at two extreme points which correspond to perfect price discrimination. On the other hand, when the cost is sufficiently high ($k=5$ case), there is no gain from segmentation because $V(\mu)$ and $v(\mu)$ overlap. For a moderate level of cost ($k=0.8$ case), the optimal segments occur at the two points where $V(\mu)$ is tangent to $v(\mu)$. Here, we can gain important insights about optimal segments. Suppose the optimal segments in the ($k=0.8$ case) are called $\mu^{1}$ at the left tangency point and $\mu^{2}$ at the right tangency point. One can easily see that any $\mu^{*}$ between $\mu^{1}$ and $\mu^{2}$ would give the same optimal segments. This is called as  local invariance in the literature. If your aggregate market $\mu^{*}$ fits in a region where there is a gain from segmentation, then the optimal segments are fixed. The only thing that changes is the fraction of each segment, and Bayes plausibility conditions check the changes in fractions of segments.

\section{Welfare Analysis} 
Understanding how consumers are affected by privacy regulations in the market is highly important for policymakers. In the following Proposition and in Figure 3, I show that restricting information access can hurt consumers. In other words, in some cases, when the monopolist has easy access to consumer information, it benefits consumers. A similar situation occurs when Amazon has your student information and charges you affordable prices.
\label{sec:analysis}

\begin{proposition}
\label{prop:3}
When \( |\Omega|=2 \), the consumer surplus (CS) always increases with the cost level when the initial market distribution \( \mu^{*} \) belongs to \( \mathcal{M}_{1} \). However, when \( \mu^{*} \) is in \( \mathcal{M}_{2} \), the CS exhibits a non-monotonic structure with increasing cost level.
\end{proposition}
Proposition 4 addresses how social welfare changes, and as seen in Figure 4, the non-monotonicity result holds in $\mathcal{M}_1$ this time. When the cost level is extremely low or high, we observe the highest social welfare. However, for moderate levels, there will be some efficiency loss due to some low types being placed in the high segment.
\begin{figure}
    \centering
    \includegraphics[scale=0.7]{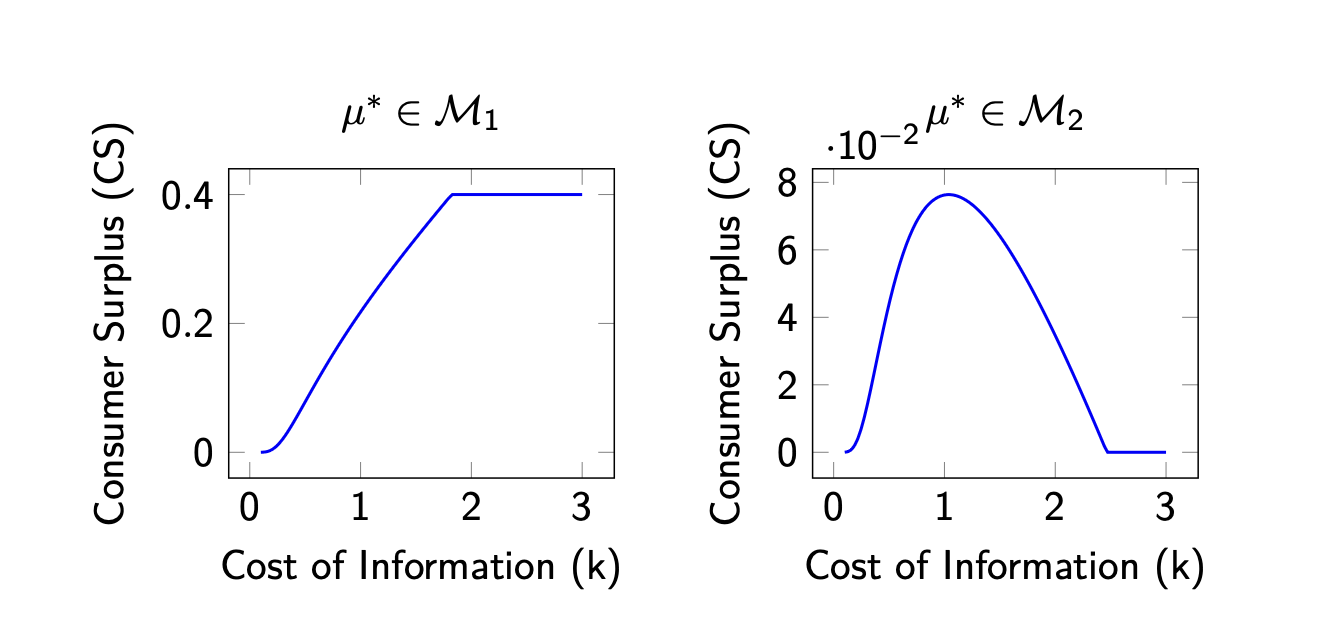}
    \label{fig:CS_Analysis}
    \caption{Consumer Surplus vs Cost of Information for different market distributions}
\end{figure}
\newpage 
\begin{proposition}
\label{prop:4}
When \( |\Omega|=2 \), the total surplus (TS) always decreases with the cost level when the initial market distribution \( \mu^{*} \) belongs to \( \mathcal{M}_{2} \). However, when \( \mu^{*} \) is in \( \mathcal{M}_{1} \), the TS exhibits non-monotonic structure with increasing cost level.
\end{proposition}
\begin{figure}
    \centering
    \includegraphics[scale=0.7]{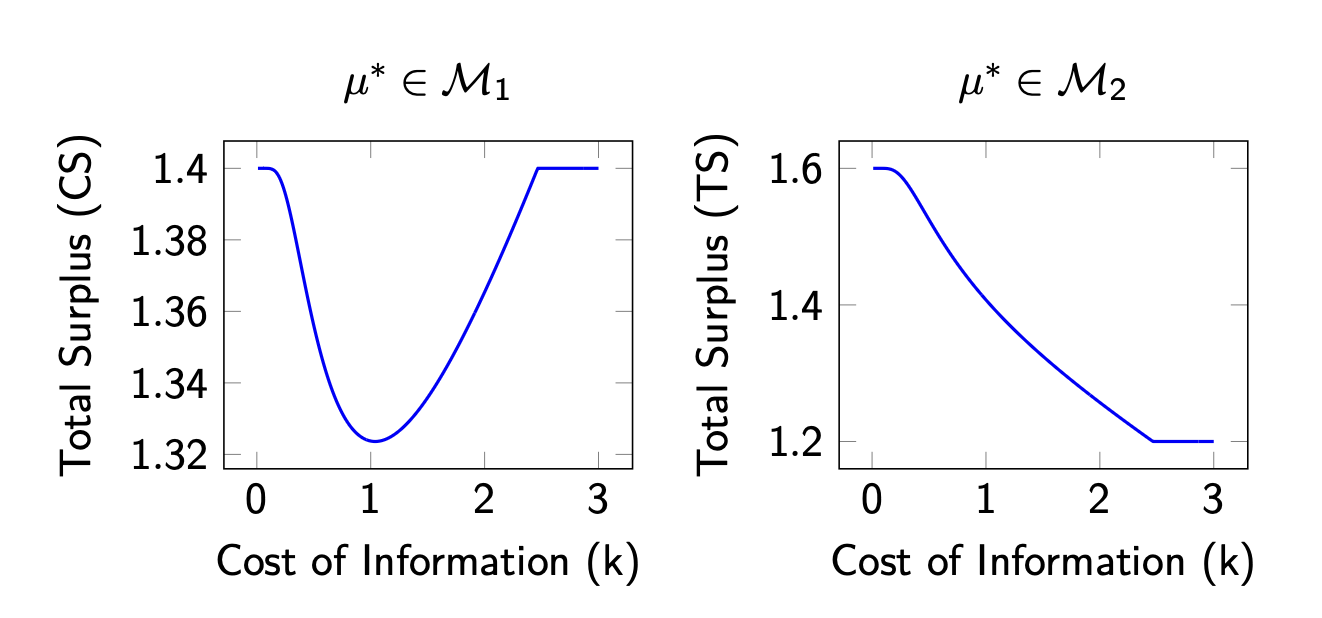}
    \label{fig:CS_Analysis}
    \caption{Total Surplus vs Cost of Information for different market distributions}
\end{figure}
\section{Rationalization}
\label{sec:rationalization}
So far, I have assumed that the monopolist has knowledge about the cost of information acquisition in the market.
Based on that assumption, I solved the monopolist's problem and conducted an analysis regarding the welfare implications of that cost on the market. My analysis primarily relied on using Shannon entropy as the cost function. However, it would be beneficial to explore more general cost functions while preserving some important structural characteristics such as strict convexity and posterior-separability, rather than focusing solely on Shannon entropy.

\cite{BBM} introduced the concept of a \textit{surplus triangle}, which represents the combination of all (CS, PS) pairs that can be achieved via some segmentation in the market. The surplus triangle introduced by \cite{BBM} does not account for cost of market segmentation. When the cost is introduced to the model, a specific cost function and fixed cost parameter $k$ determine a single point within that surplus triangle. As k changes, it creates a curved line inside the surplus triangle. The selection of different cost functions results in distinct points within the surplus triangle. 
\begin{figure}[h]
    \centering
    \begin{subfigure}[b]{0.45\textwidth} 
        \centering
        \includegraphics[width=\textwidth]{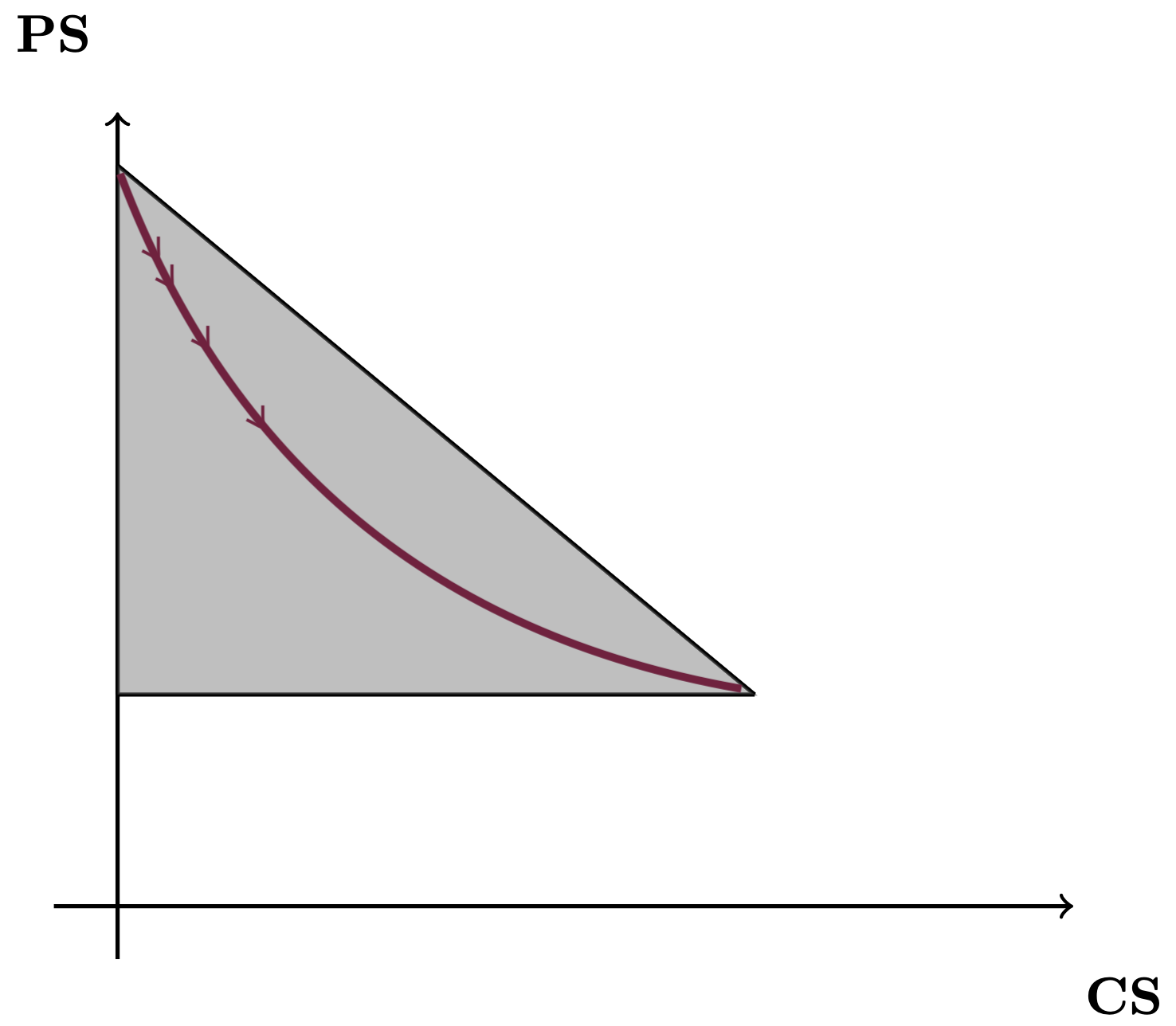}
        \caption{$\mu^{*} \in \mathcal{M}_1$}
        \label{fig:m1}
    \end{subfigure}
    \begin{subfigure}[b]{0.45\textwidth} 
        \centering
        \includegraphics[width=\textwidth]{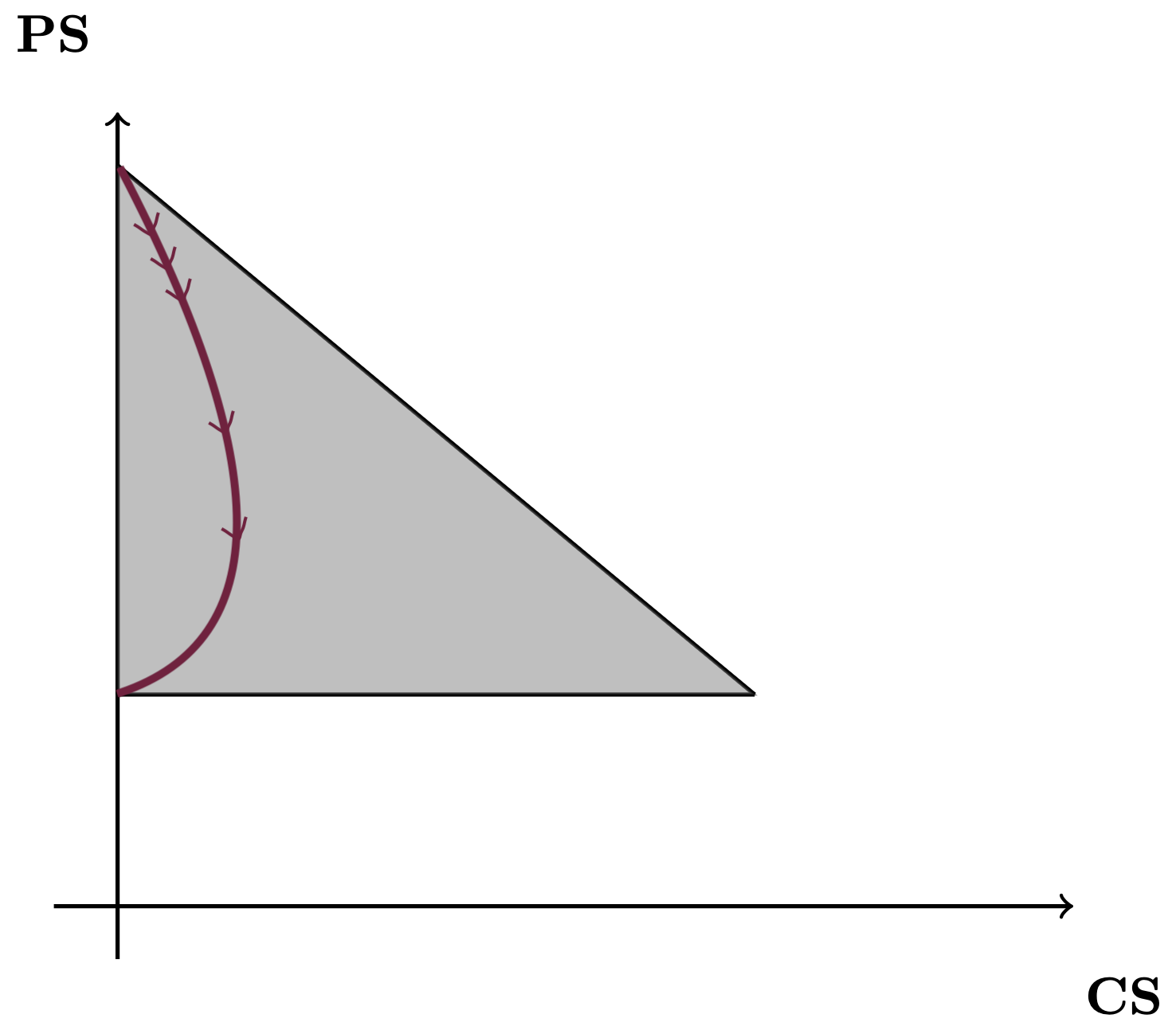}
        \caption{$\mu^{*} \in \mathcal{M}_2$}
        \label{fig:m2}
    \end{subfigure}
    \caption{The figure captures all (CS,PS) for optimal segmentation with increasing level of cost. Arrows on the curvature represent the increase on the cost level. }
    \label{fig:optimal_segmentation}
\end{figure}
The key issue is identifying the set of all attainable points within the surplus triangle resulting from the monopolist's optimization problem. This set includes all (CS, PS) pairs attainable under cost functions satisfying strict convexity and posterior-separability. I will refer to such attainable (CS, PS) pairs as rationalizable, as defined below.

\begin{definition}

A (CS, PS) pair is considered \textbf{rationalizable} if there exists a segmentation $\tau$ and a strictly convex and posterior-separable cost function $c(\cdot)$ such that the segmentation $\tau$ is optimal for the monopolist under the cost function $c(\cdot)$, and it yields the given (CS, PS) pair within the surplus triangle.
\end{definition}
\begin{proposition}
\label{prop:5}
When there are two types of consumers, any interior points of surplus triangle for a given aggregate market $\mu^*$ is rationalizable.
\end{proposition}
\begin{proof}
Detailed proof is in Appendix.
\end{proof}

\section{Conclusion} \label{sec:conclusion}
In this paper, I analyzed third-degree price discrimination under costly information acquisition. By using expected reduction of entropy as a cost function for monoplist, I provided a solution for optimal market segmentation. I also conducted welfare analysis by showing that consumer surplus has a non-monotonic structure with the cost of information. My results are mostly based on the binary case, but further research can enhance them for general cases.

\setlength\bibsep{10pt} 
\bibliographystyle{apalike}

\clearpage

\setlength{\bibsep}{0pt} 
\onehalfspacing 

\clearpage

\section{Appendix} \label{sec:appendixa}
\addcontentsline{toc}{section}{Appendix A}

\subsection{Proof of Lemmas}
\begin{lemmap}
Towards a contradiction, let $\tau^*$ be the optimal segmentation and assume that $|supp(\tau^*)| > |\Omega| = K$. By the Pigeonhole Principle, there must be at least two segments in the support of $\tau^*$ that belong to the same price region. Without loss of generality, suppose $\mu^i$ and $\mu^j \in supp(\tau^*)$ are both in $\mathcal{M}_1$.

Since $\tau^*$ is the optimal segmentation, the expected profit for the monopolist using $\tau^*$ is maximized. Let $V(\tau^*)$ denote the expected profit under $\tau^*$.

Now, consider a new segmentation $\tau'$, which is the same as $\tau^*$ except that it combines the two segments $\mu^i$ and $\mu^j$ into a single segment. Let $\mu^{ij}$ be the combined segment.

We need to show that the expected profit under $\tau'$ is strictly higher than the expected profit under $\tau^*$. Since $\mu^i$ and $\mu^j$ are in the same $\mathcal{M}_1$: their contributions to the profit function without cost are linear, and in the same price region, the combined segment would give the same expected utility. However, since the cost function is convex, Jensen's inequality proves that the cost of the combined segment is less than the expected cost of two segments. Therefore, the overall new segmentation $\tau'$ will give strictly higher expected profit to the monopolist, which contradicts the optimality of the initial segmentation. Thus, we conclude that $|supp(\tau^*)| \leq |\Omega| = K$.
\end{lemmap}

\subsection{Proof of Propositions}

\begin{proofp}

The net utility \( v(\mu^j) \) is given by:

\[
v(\mu^j) =  \sum_{i=1}^K \mu_i^j s(\omega_j, \omega_i) + k H(\mu^j)
\]

where the entropy \( H(\mu^j) \) is:

\[
H(\mu^j) = -\sum_{i=1}^K \mu_i^j \ln \mu_i^j
\]

Therefore,
\[
v(\mu^j) =  \sum_{i=1}^K \mu_i^j (s(\omega_j, \omega_i) -k \ln\mu_i^j)
\]

Lagrangian would be as follows:

\[
\mathcal{L} = \sum_j \tau(\mu^j) \left( \sum_{i=1}^K \mu_i^j (s(\omega_j, \omega_i) -k \ln\mu_i^j)\right) - \lambda \left( \sum_j \tau(\mu^j)\mu^{j} - \mu^{*} \right)
\]

Because the entropy function has a natural logarithm (ln) and its derivative goes to infinity at the boundaries, all chosen segments must have full support, i.e., for all chosen segments $\mu^{j}$, $\mu^{j}_{i} > 0$.To find the optimal solution, we take the partial derivatives of \( \mathcal{L} \) with respect to \( \mu_i^j \) and \( \lambda \), and set them to zero.

The net marginal benefit of each type should be equal to each other for each chosen segment. Let $\mu^{j}$ and $\mu^{m} \in supp(\tau^{*})$, then for any type i,

\[s(\omega_j,\omega_i) - k \ln \mu^{j}_{i} = s(\omega_m,\omega_i) - k\ln \mu^{m}_{i}\]

\[\implies k(\ln \mu^{m}_{i} - \ln \mu^{j}_{i}) = s(\omega_m,\omega_i) - s(\omega_j,\omega_i)\]

\[\implies k\left(\frac{\mu^{m}_{i}}{\mu^{j}_{i}}\right) = \exp\left(s(\omega_m,\omega_i) - s(\omega_j,\omega_i)\right)\]

\[\implies \frac{\mu_{i}^{j}}{e^{\frac{s(\omega_j,\omega_{i})}{k}}} = \frac{\mu_{i}^{m}}{e^{\frac{s(\omega_m,\omega_{i})}{k}}}\]

For the third part of the proof and extra details \cite{caplin_dean} provides detailed explanation.
 
\end{proofp}

\begin{proofp}
 
Since the monopolist is indifferent between $\omega_r$ and $\omega_s$, we can write the monopolist's profit before segmentation as:

\begin{equation*}
\omega_r \sum_{i=r}^{K}\mu^{*}_i=\omega_s \sum_{i=s}^{K}\mu^{*}_i=max_{p} \sum_{i=1}^{K} s(p,\omega_{i}) \mu^{*}_{i}
\end{equation*}

Now, assume the monopolist segments the market in the following way: she chooses two segments with equal likelihood, one in $\mathcal{M}_r$ and one in $\mathcal{M}_s$, and the segments are $\mu^{r}=(\mu^{*}_1,..,\mu^{*}_r+\epsilon,..,\mu^{*}_s-\epsilon,..\mu^{*}_K)$ and $\mu^{s}=(\mu^{*}_1,..,\mu^{*}_r-\epsilon,..,\mu^{*}_s+\epsilon,..\mu^{*}_K)$ such that $\frac{1}{2}\mu^{r}+\frac{1}{2}\mu^{s}=\mu^{*}$

After segmentation, the monopolist gains $\frac{1}{2}\omega_s \epsilon$, and the cost she pays is equal to the expected reduction in entropy, which is equal to $k(H(\mu^{*})-\frac{1}{2}(H(\mu^{r})+H(\mu^{s})))$.

We need show that the gain is always higher than the cost for sufficiently small $\epsilon$

Given that \(\mu^r\) and \(\mu^s\) only differ from \(\mu^*\) in the \(\epsilon\) perturbations, we can write:

\[
H(\mu^r) = - \left( \sum_{i \ne r,s} \mu_i^* \ln \mu_i^* + (\mu_r^* + \epsilon) \ln (\mu_r^* + \epsilon) + (\mu_s^* - \epsilon) \ln(\mu_s^* - \epsilon) \right)
\]
\[
H(\mu^s) = - \left( \sum_{i \ne r,s} \mu_i^* \ln \mu_i^* + (\mu_r^* - \epsilon) \ln (\mu_r^* - \epsilon) + (\mu_s^* + \epsilon) \ln (\mu_s^* + \epsilon) \right)
\]

So the cost for segmentation is equal to:
\begin{align*}
& \frac{k}{2} \left( - \mu_r^* \ln \mu_r^* - \mu_s^* \ln \mu_s^*  + (\mu_r^* + \epsilon) \ln (\mu_r^* + \epsilon) + (\mu_r^* - \epsilon) \ln (\mu_r^* - \epsilon) + \right.\\
& \left(\mu_s^* + \epsilon) 
\ln(\mu_s^* + \epsilon) + (\mu_s^* - \epsilon) \ln (\mu_s^* - \epsilon) \right)
\end{align*}

We want to prove that for sufficiently small $\epsilon$,
\begin{align*}
&\frac{k}{2} \left( - \mu_r^* \ln \mu_r^* - \mu_s^* \ln \mu_s^* + (\mu_r^* + \epsilon) \ln (\mu_r^* + \epsilon) + (\mu_r^* - \epsilon) \ln (\mu_r^* - \epsilon) +\right.\\
& \left (\mu_s^* + \epsilon) \ln(\mu_s^* + \epsilon) + (\mu_s^* - \epsilon) \ln (\mu_s^* - \epsilon) \right)
\end{align*}
is always less than $\frac{1}{2}\omega_s \epsilon$. 
We start by using the Taylor expansion of the $\ln()$ function around $\mu_r^*$ and $\mu_s^*$. For small $\epsilon$, the Taylor expansions are:
\begin{align*}
\ln(\mu_r^* + \epsilon) &\approx \ln \mu_r^* + \frac{\epsilon}{\mu_r^*} - \frac{\epsilon^2}{2(\mu_r^*)^2}, \\
\ln(\mu_r^* - \epsilon) &\approx \ln \mu_r^* - \frac{\epsilon}{\mu_r^*} +\frac{\epsilon^2}{2(\mu_r^*)^2}, \\
\ln(\mu_s^* + \epsilon) &\approx \ln \mu_s^* + \frac{\epsilon}{\mu_s^*} - \frac{\epsilon^2}{2(\mu_s^*)^2}, \\
\ln(\mu_s^* - \epsilon) &\approx \ln \mu_s^* - \frac{\epsilon}{\mu_s^*} +\frac{\epsilon^2}{2(\mu_s^*)^2}.
\end{align*}

Substituting these expansions into the original expression, we get:
\begin{align*}
&\frac{k}{2} \left( - \mu_r^* \ln \mu_r^* - \mu_s^* \ln \mu_s^* + (\mu_r^* + \epsilon) \ln (\mu_r^* + \epsilon) + (\mu_r^* - \epsilon) \ln (\mu_r^* - \epsilon) +\right.\\
& \left (\mu_s^* + \epsilon) \ln(\mu_s^* + \epsilon) + (\mu_s^* - \epsilon) \ln (\mu_s^* - \epsilon) \right) \\
&= \frac{k}{2} \left( - \mu_r^* \ln \mu_r^* - \mu_s^* \ln \mu_s^* + (\mu_r^* + \epsilon) \left( \ln \mu_r^* + \frac{\epsilon}{\mu_r^*} - \frac{\epsilon^2}{2(\mu_r^*)^2} \right) + (\mu_r^* - \epsilon) \left( \ln \mu_r^* - \frac{\epsilon}{\mu_r^*} + \frac{\epsilon^2}{2(\mu_r^*)^2} \right) \right. \\
&\quad + (\mu_s^* + \epsilon) \left( \ln \mu_s^* + \frac{\epsilon}{\mu_s^*} - \frac{\epsilon^2}{2(\mu_s^*)^2} \right) + (\mu_s^* - \epsilon) \left( \ln \mu_s^* - \frac{\epsilon}{\mu_s^*} + \frac{\epsilon^2}{2(\mu_s^*)^2} \right) \Bigg) \\
&= \frac{k}{2} \left( \mu_r^* \ln \mu_r^* + \mu_s^* \ln \mu_s^* + \frac{2 \epsilon^2}{\mu_r^*} + \frac{2 \epsilon^2}{\mu_s^*} \right)
\end{align*}

Since $ln(\mu)$ is negative for $\mu<1$, we have 

\[\frac{k}{2} \left( \mu_r^* \ln \mu_r^* + \mu_s^* \ln \mu_s^* + \frac{2 \epsilon^2}{\mu_r^*} + \frac{2 \epsilon^2}{\mu_s^*} \right)<\frac{k}{2} \left(\frac{2 \epsilon^2}{\mu_r^*} + \frac{2 \epsilon^2}{\mu_s^*} \right)=\epsilon^2 \left(\frac{ k}{\mu_r^*} + \frac{k}{\mu_s^*} \right)
\]

For sufficiently small $\epsilon$,

\[
\epsilon^2 \left(\frac{ k}{\mu_r^*} + \frac{k}{\mu_s^*} \right) < \frac{1}{2}\omega_s \epsilon.
\]

Therefore, the monopolist always chooses to segment the market, regardless of the cost level, when she is indifferent between charging two prices.

\end{proofp}

\begin{proofp}
W.L.O.G, we can set the price $\omega_1=1$ as  numéraire  and we can set $\omega_2=\omega$ where $\omega>1$ because only thing matters is the ratio of two valuations.
        
If we follow Proposition 1, in the optimal segmentation:
        \begin{align*}
            \mu^1_{2} &= \displaystyle\frac{e^{\frac{1}{k}}-1}{e^{\frac{\omega}{k}}-1} \\
            \mu^2_{2} &= \displaystyle\frac{e^{\frac{\omega}{k}}-e^{\frac{\omega-1}{k}}}{e^{\frac{\omega}{k}}-1} \\
            \tau(\mu^1) &= \displaystyle\frac{e^{\frac{\omega}{k}}(1-\mu^{*})+\mu^{*}-e^{\frac{\omega-1}{k}}}{(e^{\frac{\omega-1}{k}}-1)(e^{\frac{1}{k}}-1)}
        \end{align*}
Let:
\[
A = e^{\frac{\omega}{k}}, \quad B = e^{\frac{\omega - 1}{k}}
\]
We can calculate consumer surplus (CS) as follows:
\[
CS(\omega, k, \mu^*) = \tau(\mu^1)\mu^1_{2}(\omega-1)=\frac{A(1 - \mu^*) + \mu^* - B}{(B - 1)(A- 1)}(\omega-1)
\]
and we need to show that $CS$ always increases with $k$ regardless of $\omega$ when $\mu^* \in \mathcal{M}_1$ and it is segmentable, i.e., $\frac{1}{\omega} \geq \mu^*> \mu^1_{2}= \frac{A-B}{B(A-1)}$

To show that $CS$ increases with $k$, we need to show that the partial derivative of $CS$ with respect to $k$ is always positive.

\[
CS(\omega, k, \mu^*) = \frac{A (1 - \mu^*) + \mu^* - B}{(B - 1) (A - 1)}(\omega-1)
=(\frac{1-\mu^{*}}{B-1}-\frac{1}{A-1})(\omega-1)\]
We can ignore \(\omega-1\) when taking the partial derivative of \(CS\) with respect to \(k\) because we only consider the sign of the partial derivative. Since \(\omega > 1\), that part wouldn't affect the sign of the derivative.

\[
\frac{\partial CS}{\partial k} = \frac{\partial}{\partial k} \left(\frac{1-\mu^{*}}{B-1}-\frac{1}{A-1}\right)
\]

We have the following partial derivatives:
\[
\frac{\partial A}{\partial k} = -\frac{\omega}{k^2} e^{\frac{\omega}{k}} = -\frac{\omega}{k^2} A
\]
\[
\frac{\partial B}{\partial k} = -\frac{\omega - 1}{k^2} e^{\frac{\omega - 1}{k}} = -\frac{\omega - 1}{k^2} B
\]

Since $  \frac{A-B}{B(A-1)}< \mu^* \leq \frac{1}{\omega} $, we know that $\frac{A(B-1)}{B(A-1)}>1 - \mu^* \geq 1 - \frac{1}{\omega}$. 
When we substitute the partial derivatives into the derivative of $CS$ and using $\frac{\ln(A)}{\ln(B)}=\frac{\omega}{\omega-1}$ we get:

\[\frac{\partial CS}{\partial k} = \frac{(1-\mu^{*})(\omega - 1) B}{k^2 (B-1)^2} - \frac{\omega A}{k^2 (A-1)^2} > 0 
\]
\begin{align*}
\iff \frac{(1-\mu^{*})(\omega - 1) B}{k^2 (B-1)^2} &> \frac{\omega A}{k^2 (A-1)^2} \\
\iff \frac{(1-\mu^{*})(\omega - 1) B}{(B-1)^2} &> \frac{\omega A}{(A-1)^2}\\
\iff \frac{(1-\mu^{*})\ln(B) B}{(B-1)^2} &> \frac{\ln(A) A}{(A-1)^2}\\
\end{align*}

We know $1-\mu^* \geq 1 - \frac{1}{\omega}$ which means  $1-\mu^* \geq  \frac{\ln(B)}{\ln(A)}$. So
\[
\frac{(1-\mu^{*})\ln(B) B}{(B-1)^2}> \frac{\ln^2(B) B}{\ln(A)(B-1)^2} 
\]

It is enough to show that:
\begin{align*}
 \frac{\ln^2(B) B}{\ln(A)(B-1)^2} &>\frac{\ln(A) A}{(A-1)^2}\\
 \iff  \frac{\ln^2(B) B}{(B-1)^2} &> \frac{\ln^2(A) A}{(A-1)^2}
\end{align*}

Let:
\[
f(x)=\displaystyle\frac{\ln^2(x)x}{(x-1)^2}
\]

It is enough to show that $f(x)$ is decreasing when $x>1$ because we know that $A>B>1$. One can easily prove that $f(x)$ is indeed a decreasing function when $x>1$. Therefore, we have shown that $CS$ always increases with $k$ regardless of $\omega$ when $\mu^{*} \in \mathcal{M}_1$. The proof of non-monotonicity in $\mathcal{M}_2$ is relatively easy. When $k=0$, the monopolist chooses the PPD, which means CS is zero. As $k$ goes to infinity, the monopolist chooses uniform monopoly pricing, in which CS is again 0. However, CS is a continuous function of $k$ in the binary case. If a continuous function takes 0 values at two ends, and we know the existence of a segment in the low price region which brings a positive surplus for the high type and positive value for CS, then CS must be non-monotonic.
 
\end{proofp}

\begin{proofp}

 This time, it is easy to prove that TS is non-monotonic in $\mathcal{M}_1$. When $k=0$, the monopolist chooses the PPD, which means TS is maximized. As $k$ goes to infinity, the monopolist chooses uniform monopoly pricing, in which TS is again maximized because all types are involved in trade, making it an efficient point. However, TS is a continuous function of $k$ in the binary case. If a continuous function takes maximum values at two ends, and we know the existence of segments in the high price region which exclude some low type people from the market and cause inefficiency, then TS must be non-monotonic. The monotonicity result follows the similar logic as I used in the proof of Proposition 3.   
\end{proofp}

\begin{proofp}

Let's assume there are only two types of consumers in the market, with valuations $\omega_1$ and $\omega_2$ such that $\omega_1 < \omega_2$. Without loss of generality, let's also assume that the initial aggregate market $\mu^*$ is given and $\mu^* \in \mathcal{M}_1$, such that $p^{*}(\mu^{*})=\omega_1$ is optimal for the monopolist in the initial aggregate market $\mu^*$. 

The surplus triangle in this case would be as follows, where the red dot represents the uniform monopoly pricing case:

\begin{tikzpicture}
[scale=5]
  \draw[->] (-0.2,0) -- (1.4,0) node[below] {Consumer Surplus} ;
  \draw[->] (0,-0.2) -- (0,1.4) node[above] {Producer Surplus} ;
  \draw (0,0.5) -- (0.8,0.5) -- (0,1) -- cycle;
  \draw[line width=2.5pt] (0,0.5) -- (0.8,0.5) -- (0,1) -- cycle;
  \fill[gray] (0,0.5) -- (0.8,0.5) -- (0,1) -- cycle;
  \node[below left] at (0,0.5) {$\omega_1$};
  \node[below right] at (0.8,0.5) {$\mu^{*}(\omega_2-\omega_1)$};
  \filldraw[red] (0.8,0.5) circle (0.4pt);
  \node[above left] at (0,1) {$(1-\mu^{*})\omega_{1}+\mu^{*}\omega_2$};
\end{tikzpicture}

With two segments $\mu^1 \in [0,\mu^*]$\footnote{Here, $\mu^*$ represents the ratio of type two consumers in the market, and $\mathcal{M}_1$ corresponds to the interval $[0,\frac{\omega_1}{\omega_2}]$.} and $\mu^2 \in [\frac{\omega_1}{\omega_2},1]$ such that $\tau( \mu^1) + \tau(\mu^2) = \mu^*$, we have:

$CS=\tau(\mu^{1})(\omega_2-\omega_1)$ and $PS=\tau(\mu^{1})\omega_1+\tau(\mu^{2})\mu^{2}\omega_{2}$.\footnote{When I perform welfare analysis, I allow the monopolist to optimize her decision by taking the cost into consideration, but I analyze PS in the optimal segments with the cost part excluded.}

Since $0 \leq \tau(\mu^{1}) \leq \mu^{*}$, $CS$ can take any value in the interval $[0, \mu^{*}(\omega_2 - \omega_1)]$. For given $CS$ and $PS$ values (and given $\omega_{1}$ and $\omega_{2}$), one can solve $\mu^1$ and $\mu^2$ uniquely using the $CS$ and $PS$ equations because the only unknown parameters are $\mu^{1}$ and $\tau(\mu^{1})$.

So, a given $(CS, PS)$ pair induces optimal segments in the market.\footnote{This is not the case in the general situation where $|\Omega| \geq 3$ because, in that case, the same $(CS, PS)$ pair could be produced via different segmentations.
}

We can find the conditions that optimal segmentation should satisfy by solving the monopolist's problem over two segments. Let's assume $\mu^{*1}$ and $\mu^{*2}$ are the \textit{induced} optimal segments for the given $(CS, PS)$ pair. The monopolist's problem can be written as:

$$
\begin{aligned}
& \underset{\{\mu^1,\mu^2\}}{\text{maximize}}
& &  \tau(\mu^{1}) \omega_1 +\tau(\mu^{2})\omega_2 \mu^2- \tau((\mu^{1}) c(\mu^1) - \tau((\mu^{2}) c(\mu^2) \\
& \text{subject to}
& & \tau(\mu^1) \mu^1 + (\tau(\mu^2)) \mu^2 = \mu^* \quad  (\tau(\mu^1) \geq 0,\tau(\mu^2) \geq 0, \text{ and } \tau(\mu^1) +\tau(\mu^2)=1)\\
\end{aligned}
$$
\newpage
F.O.C. with respect to $\mu^1$:

\begin{align*}
    \frac{\partial}{\partial \mu^1}(\tau(\mu^1) \omega_1 +(1-\tau(\mu^1)) \omega_2 \mu^2- \tau(\mu^1)c(\mu^1) - (1-\tau(\mu^1)) c(\mu^2)) = 0 \\
\end{align*}
\[ \implies
\frac{\partial \tau(\mu^1) }{\partial \mu^1}(\omega_1-\omega_2 \mu^{*2})+\frac{\partial \tau(\mu^{1}) }{\partial \mu^{1}}(c(\mu^{*2})-c(\mu^{*1}))=\tau(\mu^{*1})c'(\mu^{*1}) 
    \text{...(1)}
\]

F.O.C. with respect to $\mu^2$:

\begin{align*}
    \frac{\partial}{\partial \mu^2}(\tau(\mu^{1}) \omega_1 +(1-\tau(\mu^{1})) \omega_2 \mu^{2}- \tau(\mu^{1})c(\mu^1) - (1-\tau(\mu^{1})) c(\mu^2)) = 0 \\
\end{align*}
\[ \implies
\omega_2 =c'(\mu^{*2}) 
    \text{...(2)}
\]

We also know that:
\[\tau(\mu^{1})=\frac{\mu^{2}-\mu^{*}}{\mu^{2}-\mu^{1}}\]
\[\implies \frac{\partial \tau(\mu^1) }{\partial \mu^1}=\frac{\tau(\mu^{1})}{\mu^{2}-\mu^{1}}\]
Moreover, the net gain from segmentation should also be positive. This is a necessary condition for the monopolist to segment the market. This condition can be written as:

\[
  \tau(\mu^{*1}) c(\mu^{*1}) + \tau(\mu^{*2}) c(\mu^{*2})) \leq \tau(\mu^{*2})(\omega_{2}\mu^{*2}-\omega_{1}) \text{...(3)}
\]
Now, the question is: can we find a strictly convex and continuous function that satisfies the above conditions? The answer is yes because these are only local conditions that should be satisfied by some convex function. We can construct the function $c(.)$ in the following way:

Since we are free to select $c(.)$ except for being strictly convex, we can select our $c(.)$ such that it satisfies $c(\mu^{*1})=c(\mu^{*2})$. Then, (1) gives:

\[
\omega_{1}-\omega_{2}\mu^{*}_{2}=c^{'}(\mu^{*}_{1})(\mu^{*}_{2}-\mu^{*}_{1})\\
\implies c^{'}(\mu^{*}_{1})=\displaystyle\frac{\omega_{1}-\omega_{2}\mu^{*}_{2}}{\mu^{*}_{2}-\mu^{*}_{1}}
\]

and (2) gives:

$$
 c^{'}(\mu^{*}_{2})=\omega_2
$$

Since we know the values of $\omega_1,\omega_2,\mu^{*}_{1}$, and $\mu^{*}_{2}$, we can find a convex function $c(.)$ that satisfies the equalities above, and we can also make it satisfy (3) by adjusting the constant term in $c(.)$.

Therefore, any interior point of the surplus triangle can be uniquely represented by two interior segments, $\mu^{*1}$ and $\mu^{*2}$, which satisfy the necessary and sufficient conditions for optimal segmentation.

\end{proofp}

\end{document}